\documentclass[11pt]{article}

\usepackage{epsfig}
\usepackage{amsfonts}
\usepackage{amssymb}
\usepackage{amsthm}
\usepackage{amstext}
\usepackage{amsmath}
\usepackage{xspace}
\usepackage{graphicx}
\usepackage{fullpage}
\usepackage{graphics}
\usepackage{colordvi}
\usepackage{color}

\advance \oddsidemargin by -0.8in
\newcommand{\myparskip}{3pt}
\parskip \myparskip
\newcommand{\yi}{{\sc Yes-Instance}\xspace}
\renewcommand{\ni}{{\sc No-Instance}\xspace}

\newcommand{\ceil}[1]{\ensuremath{\left\lceil#1\right\rceil}}

\newcommand{\abs}[1]{\lvert #1\rvert}
\newcommand{\paren}[1]{\left ( #1 \right ) }

\newcommand{\NP}{\mbox{\sf NP}}
\newcommand{\APX}{\mbox{\sf APX}}

\newcommand{\opt}{\mbox{\sf OPT}}

\newcommand{\set}[1]{\left\{ #1 \right\}}
\newcommand{\sse}{\subseteq}

\newcommand{\tset}{{\mathcal T}}

\newcommand{\qset}{{\mathcal{Q}}}

\newcommand{\aset}{{\mathcal{A}}}

\newcommand{\rset}{{\mathcal{R}}}

\newcommand{\sset}{{\mathcal{S}}}

\newtheorem{theorem}{Theorem}[section]
\newtheorem{lemma}[theorem]{Lemma}

\newtheorem{corollary}[theorem]{Corollary}
\newtheorem{claim}[theorem]{Claim}

\def\square{\vbox{\hrule height.2pt\hbox{\vrule width.2pt height5pt \kern5pt
\vrule width.2pt} \hrule height.2pt}}

\theoremstyle{definition}\newtheorem{definition}[theorem]{Definition}

\newenvironment{prog}[1]{
\begin{minipage}{5.8 in}
{\sc\bf #1}
\begin{enumerate}}
{
\end{enumerate}
\end{minipage}
}

\renewcommand{\phi}{\varphi}
\newcommand{\eps}{\epsilon}

\newcommand{\expect}[2]{\text{\bf E}_{#1}\left [#2\right]}

\newcommand{\maxsatf}{\mbox{\sf Max 3SAT(5)} }
\newcommand{\threesat}{\mbox{\sf Max 3SAT} }
\newcommand{\twosat}{\mbox{\sf Max 2SAT} }

\newcommand{\len}{\operatorname{len}}
\newcommand{\rev}{\operatorname{rev}}

\newcommand{\events}{\mathcal E}

\def\nu#1{}
\def\ed#1{}
\def\bun#1{}

\usepackage{times}
\usepackage[small,compact]{titlesec}
\usepackage[small,it]{caption}

\newcommand{\squishlist}{
 \begin{list}{$\bullet$}
  { \setlength{\itemsep}{0pt}
     \setlength{\parsep}{3pt}
     \setlength{\topsep}{3pt}
     \setlength{\partopsep}{0pt}
     \setlength{\leftmargin}{1.5em}
     \setlength{\labelwidth}{1em}
     \setlength{\labelsep}{0.5em} } }
\newcommand{\squishend}{
  \end{list}  }

\newcounter{Lcount}
\newcommand{\squishlisttwo}{
\begin{list}{D\arabic{Lcount}. }
{ \usecounter{Lcount} \setlength{\itemsep}{0pt}
\setlength{\parsep}{0pt} \setlength{\topsep}{0pt}
\setlength{\partopsep}{0pt} \setlength{\leftmargin}{2em}
\setlength{\labelwidth}{1.5em} \setlength{\labelsep}{0.5em} } }

\newcommand{\squishendtwo}{
\end{list} }

\newcommand{\be}{\squishlisttwo}
\newcommand{\ee}{\squishendtwo}
\newcommand{\bi}{\squishlist}
\newcommand{\ei}{\squishend}

\title{Stackelberg Pricing is Hard to Approximate within $2-\epsilon$}

\author{
Parinya Chalermsook \thanks{Department of Computer Science,
 University of Chicago, Chicago, IL, USA. Email:
 {\tt parinya@uchicago.edu}}
\and
Bundit Laekhanukit \thanks{Department of Combinatorics and
  Optimization, University of Waterloo, ON, Canada. Email: {\tt
    blaekhan@uwaterloo.ca}}
\and Danupon Nanongkai \thanks{College of Computing,  Georgia Tech,
Atlanta, GA, USA. Email: {\tt danupon@cc.gatech.edu}}}
\date{}

\begin{document}



\maketitle

\begin{abstract}
Stackelberg Pricing Games is a two-level combinatorial pricing
problem studied in the Economics, Operation Research, and Computer
Science communities. In this paper, we consider the decade-old
shortest path version of this problem which is the first and most
studied problem in this family.

The game is played on a graph (representing a network) consisting of
{\em fixed cost} edges and {\em pricable} or {\em variable cost}
edges. The fixed cost edges already have some fixed price
(representing the competitor's prices). Our task is to choose prices
for the variable cost edges. After that, a client will buy the
cheapest path from a node $s$ to a node $t$, using any combination of
fixed cost and variable cost edges. The goal is to maximize the
revenue on variable cost edges.

In this paper, we show that the problem is hard to approximate
within $2-\epsilon$, improving the previous \APX-hardness result by
Joret [to appear in {\em Networks}]. Our technique combines the
existing ideas with a new insight into the price structure and its
relation to the hardness of the instances.
\end{abstract}


\section{Introduction}
\label{sec:intro}

A newly startup company has just acquired some links in a network.
The company wants to sell these links to a particular client, who
will buy a cheapest path from a node $s$ to a node $t$. However,
this company is not alone in the market: there are other companies
already in the market owning some links with some fixed prices. The
goal of this new company is to price its links to maximize its
profit, having the complete knowledge of the network and knowing
that the client will buy the cheapest $s$-$t$ path (which may
consist of links from many companies). Of course, if they price a
link too high, the client will switch to other links and if they
price a link too low then they unnecessarily reduce their profit.

This problem is called the {\em Stackelberg Shortest Path Game}
({\sc StackSP}) and can be defined formally as follows. We are given
a directed graph $G=(V,E)$, a source vertex $s$ and a sink vertex
$t$. The set $E$ of edges is partitioned into two sets: $E_f$, the
set of {\em fixed cost edges}, and $E_v$, the set of {\em pricable}
or {\em variable cost} edges. Each edge $e$ in $E_f$ already has
some price $p(e)$. Our task is to set a price $p(e)$ to each
variable cost edge $e$. Once we set the price, the client will buy a
shortest path from $s$ to $t$ (i.e., a path $P$ such that
$\sum_{e\in P} p(e)$ is minimized). Our goal is to maximize the
profit; i.e., maximize $\sum_{e\in P\cap E_v} p(e)$ where $P$ is the
path bought by the client. Throughout, we let $m$ denote the number
of variable cost edges.
It is usually assumed that if there are many shortest paths, the
client will buy the one that maximizes our profit.


Due to its connection to road network tolling and bilevel
programming, there is an enormous effort in understanding the
problem by means of bilevel programming \cite{LabbeMS98,
dewez2004thesis, DMS06, HeilpornLMS06, HeilpornLMS07, HLMS07,
DewezLMS08, BouhtouEM06}, finding polynomial-time solvable cases
\cite{LabbeMS98, vHoeselKMOB03, GrigorievHKUB04, vdKraaij-thesis04,
CardinalLLP05, BGvHvdKSU07, BriestHK08}, solving the problem by
heuristics~\cite{DimitriouTS08, DimitriouT09}, and approximating the
solution~\cite{RochSM05, BriestHK08, Joret08, vanHoesel08}. In this
paper, we focus on approximability of this problem. In this realm,
{\sc StackSP} is the first and the most studied problem in the
growing family of one-follower (i.e., one client) Stackelberg
network pricing games~\cite{RochSM05, BriestHK08, Joret08,
vanHoesel08, BiloGPW08, CardinalDFJLNW07, CardinalDFJNW09,
BriestHGV09}.

The Stackelberg pricing problems belong to the class of two-player
two-level optimization problems which is a subclass of the bilevel
linear programming. These problems have a rather strange structure,
and this makes the standard approximation techniques such as linear
programming seemingly inapplicable. For example, a natural LP
formulation for {\sc StackSP} (and also another version called {\sc
StackMST}) has an integrality gap of $\Omega(\log m)$. Moreover, by
using the most (and probably the only) natural upper bound for
$\opt$, one cannot obtain approximation factor better than $O(\log
m)$ \cite{RochSM05}, so the line of attacks considered in
\cite{RochSM05} and \cite{BriestHK08} cannot be pushed any further.

Proving the hardness of this problem seems to have an equally big
obstacle. In fact, the progress on the hardness side for the family
of Stackelberg pricing problems stops at small constant hardness
(\APX-hardness in \cite{Joret08, CardinalDFJLNW07} and only
\NP-hardness in \cite{BiloGPW08}). Moreover, a reduction from Unique
Coverage problem~\cite{DBLP:journals/siamcomp/DemaineFHS08}, which
proved useful for many pricing problems (including {\sc StackSP}
with multiple followers) apparently does not apply here.
In particular, for {\sc StackSP}, only \NP-hardness, strong
\NP-hardness, and \APX-hardness (with a constant as small as
$1.001$) are shown~\cite{LabbeMS98, RochSM05, Joret08}. In fact,
even for approximating the general bilevel program, only the
constant ratio can be ruled out~\cite{GassnerK09, Jeroslow85}.


We believe that an improvement to upper or lower bound of the
problem might shed some light on approximating a larger subclass of
bilevel programs, perhaps generating a new set of techniques for
attacking the whole family of Stackelberg problems. (The problem
seems to require a new technique due to its bizarre behavior.)






\paragraph{Our result and techniques}
In this paper, we give the first result beyond a very small constant
hardness:

\begin{theorem}\label{thm:maintheorem}
For any $\eps>0$, it is NP-hard to approximate {\sc StackSP} to
within a factor of $2-\eps$.
\end{theorem}

The key insight in obtaining this result comes from exploring the
structure of the edge prices which was not exploited in the previous
inapproximability results~\cite{LabbeMS98, RochSM05, Joret08}:
The previous results encode the constraints in the constraint
satisfaction problems (\threesat in their cases) using certain
gadgets and glue these gadgets together in a uniform way (i.e.,
using the same edge price throughout).
%
However, we study the influence of non-uniform prices to the
hardness of the resulting instances. In particular, we study how the
prices of the fixed cost edges affect the hardness of the gadgets
and found an optimal price which strikes a balance between being too
high (which could hugely reduce the revenue but is easy to avoid)
and too low (which is likely to be used but do not affect the
revenue much).
This observation, armed with a stronger constraint satisfaction
problem (i.e., Raz verifier for \maxsatf) and a right parameter of
price, leads to a $(2-\eps)$-hardness of approximation.
The techniques above are strong enough that the hardness result is
obtained with only a slight modification of the gadgets. However,
due to the non-uniformity of the prices, a more sophisticated
analysis is required.
In particular, our analysis relies on a technique called {\em Path
Decomposition} which breaks the shortest path in the optimal
solution into subpaths with manageable structure. We will be able to
get deeper into the intuition after we describe the hardness
construction in the next section.



%

\subsubsection*{Related work}
{\sc StackSP} is first proposed by Lab\'{b}e et al.~\cite{LabbeMS98}
who also derive a bilevel LP formulation of the problem and prove
\NP-hardness. On the algorithmic side, Roch et al. present the
first, and still the best, approximation algorithm which attains
$O(\log m)$ approximation factor. Another $O(\log m)$ approximation
algorithm is obtained by Briest et al, which has a slightly worse
approximation guarantee (larger constant in front of $\log m$ term)
but is simpler and applicable to a much richer class of Stackelberg
pricing problems. Even though the algorithm of Briest et al. does
not rely on the specific problems' structures, it remains unclear
whether one can exploit a special structure of each problem to
improve the approximation ratio.

Another interesting problem in the family of one-follower
Stackelberg network games is {\em Stackelberg Minimum Spanning Trees
Game} ({\sc StackMST}) in which the client aims to buy the minimum
spanning tree instead of the shortest path. Cardinal et
al.~\cite{CardinalDFJLNW07} introduce this problem and prove that it
is \APX-hard but has an $O(\log m)$ approximation algorithm. Very
recently, they consider the special cases of planar and
bounded-treewidth graphs~\cite{CardinalDFJNW09} and prove that even
in such graph classes, {\sc StackMST} remains \NP-hard.
There are also many other variations in the family of Stackelberg
games, depending on what the client wants to buy. This includes
vertex cover~\cite{BriestHK08, BriestHGV09}, shortest path
tree~\cite{BiloGPW08}, and knapsack~\cite{BriestHGV09}.

Among the known approximation algorithms, the most universal one is
an  $O((1+\eps)\log m)$-approximation algorithm invented by Briest
et al.~\cite{BriestHK08}. This elegant algorithm works on a large
class of problems, including {\sc StackSP} and {\sc StackMST} and is
coupled with a simple analysis. In the same paper, the case of $k$
clients is also considered. An $O((1+\eps)(\log m+\log
k))$-approximation algorithm is given, and the problem is shown to
be hard to approximate within $O(\log^\eps m+\log^\eps k)$ for some
large $k$. Therefore, the gap is almost closed in the case of many
clients while left wide open when $k$ is small (e.g. $k$ is
constant, and particularly when $k=1$).

In Economics and Operation Research literature, {\sc StackSP} is
also known as a tarification problem. Many special cases are
considered and polynomial-time algorithms are given for this
problem~\cite{LabbeMS98, vHoeselKMOB03, GrigorievHKUB04,
vdKraaij-thesis04, CardinalLLP05, BGvHvdKSU07, BriestHK08}. It is
also sometimes called a bilevel pricing problem due to its
connection to the bilevel linear program. (See a formulation in,
e.g., \cite{LabbeMS98}.) {\sc StackSP} is also heavily studied from
this perspective~\cite{LabbeMS98, dewez2004thesis, DMS06,
HeilpornLMS06, HeilpornLMS07, HLMS07, DewezLMS08, BouhtouEM06}.
Approximating a solution of bilevel program to within any constant
factor is shown to be \NP-hard~\cite{Jeroslow85, GassnerK09}.
Unfortunately, these reductions do not extend to the family of
Stackelberg games due to specific structures of the constraints used
in the reduction of \cite{Jeroslow85, GassnerK09}. For more details,
we refer the readers to \cite{vanHoesel08, DewezLMS08, GassnerK09}
and references therein.

%
\paragraph{Remark} Recently Briest and Khanna~\cite{BK09} discover a
similar result to ours using a different approach. They show that
{\sc StackSP} is hard to approximate within a factor of $2-o(1)$.

\paragraph{Organization}

Our construction is a reduction from Raz verifier for \maxsatf. We
first give an overview of Raz verifier in Section~\ref{sec:prelim}.
We then describe our reduction in Section~\ref{sec:reduction} before
we are able to give more intuition behind the construction and its
analysis. This will be done in Section~\ref{sec:intuition}. We then
show a formal analysis in Section~\ref{sec:analysis}.

\section{Raz Verifier}
\label{sec:prelim}

Our reduction uses the Raz verifier for \maxsatf with $\ell$
repetitions. We explain this framework in this section.
%
The given instance of \maxsatf is a 3CNF formula with $n$
variables and $5n/3$ clauses where each clause contains exactly
$3$ different literals, and each variable appears in exactly $5$
different clauses.

Let $\eps$ be a constant and let $\phi$ be an instance of \maxsatf.
Then $\phi$ is called a \yi if there is an assignment that satisfies
all the clauses, and it is called a \ni if any assignment satisfies
at most $(1-\eps)$-fraction of the clauses. The following is a form
of the PCP theorem.

\begin{theorem}
\label{theorem: PCP theorem} There is a constant $\eps: 0 < \eps
<1$, such that it is NP-hard to distinguish between \yi and \ni of
the \maxsatf problem.
\end{theorem}



Raz verifier for \maxsatf with $\ell$ repetitions is a two-provers
one-round interactive proof system. The verifier sends one query to
each prover simultaneously. The first prover is asked for an
assignment to the variables in the given clauses while the second
prover is asked for an assignment of the variables that satisfies
all the given clauses. The verifier will accept the answers if and
only if both provers return consistent assignments. The detailed
description of the provers-verifier actions is as follows.

\bi

\item The verifier first chooses $\ell$ clauses, say
$C_1,\ldots,C_\ell$, independently and uniformly at random (with
replacement). Next, choose one variable in each of these clauses
uniformly at random. Let $x_1,\ldots,x_\ell$ denote the resulting
(not necessarily distinct) variables.

\item The verifier generates a query $q$ consisting of the indices of
  $C_1,C_2,\ldots,C_\ell$ and a query $q'$ consisting of the indices of
  $x_1,x_2,\ldots,x_\ell$. The verifier then sends $q$ and $q'$ to
  Prover~1 and Prover~2, respectively.

\item Prover~1 returns an assignment to all variables associated with
  clauses $C_1,C_2,\ldots,C_\ell$.

\item Prover~2 returns an assignment to variables
  $x_1,x_2,\ldots,x_\ell$.

\item The verifier reads the assignment received from both provers and
  accepts if and only if the assignments are consistent and satisfy
  $C_1,C_2,\ldots,C_\ell$.

\ei

Intuitively, for the \yi, both provers can ensure that the verifier
always accepts by returning the satisfying assignments to the
prover. On the other hand, any provers' strategy fails with high
probability in the case of \ni. This is an application of the
Parallel Repetition Theorem and Theorem~\ref{theorem: PCP theorem}
and can be stated formally as follows.

\begin{theorem}[~\cite{Raz98,AroraLMSS98}]
\label{theorem: Raz} There exists a universal constant $\alpha> 0$
(independent of $\ell$) such that

\bi

\item If $\phi$ is a \yi, then there is a strategy of the provers
that makes the verifier accepts with probability $1$.

\item If $\phi$ is a \ni, for any provers' strategy, the verifier
will accept with probability at most $2^{-\alpha \ell}.$
\ei
\end{theorem}

In our reduction, we view Raz verifier as the following constraint
satisfaction problem. We have two sets of queries, $\qset_1$ and
$\qset_2$, corresponding to all possible queries sent to Prover 1
and Prover 2, respectively. That is, $\qset_1$ consists of all
possible choices of $\ell$ clauses sent to Prover 1 (hence,
$|\qset_1|=(5n/3)^\ell$) and $\qset_2$ consists of all possible
choices of $\ell$ variables sent to Prover 2 (hence
$|\qset_2|=n^\ell$). For each $q\in \qset_1\cup\qset_2$, let $A(q)$
denote the set of all possible answers to $q$. Notice that
$|A(q)|=7^\ell$ if $q\in \qset_1$ (since there are $7$ ways to
satisfy each of the $\ell$ clauses given to Prover 1) and
$|A(q)|=2^\ell$ if $q\in \qset_2$ (since there are 2 possible
assignment to each of the $\ell$ variables given to Prover 2).
Denote by $\aset_1$ and $\aset_2$ the set of all possible answers by
Prover 1 and Prover 2, respectively.

We denote the set of constraints by $\Phi$. Each constraint in
$\Phi$ corresponds to a pair $(q_1,q_2)$ of queries sent by the
verifier. That is, for each random string $r$ of the verifier,
there is a constraint $(q_1, q_2)\in \qset_1\times \qset_2$ in
$\Phi$ where $q_1$ and $q_2$ are queries sent to Prover~1 and
Prover~2 respectively. A constraint $(q_1, q_2)$ is satisfied if
and only if the assignments to $q_1$ and $q_2$ are consistent. For
convenience, we will treat $\Phi$ as the set of all possible
random strings, and we denote, for each random string $r$, the
corresponding queries by $q_1(r)$ and $q_2(r)$ respectively. Note
that each query $q \in \qset_1$ is associated with $3^{\ell}$
constraints in $\Phi$ and each query $q' \in \qset_2$ with
$5^{\ell}$ constraints. Moreover, let $M= \abs{\Phi}$. We have
$M=(5n)^{\ell}$. The goal of this problem is to find an assignment
$f: \qset_1 \rightarrow \aset_1, \qset_2 \rightarrow \aset_2$ that
maximizes the number of satisfied constraints in $\Phi$.


The following corollary can be directly obtained from
Theorem~\ref{theorem: Raz}.

\begin{corollary}
If $\phi$ is a \yi, then there is an assignment to $\qset_1 \cup
\qset_2$ such that all constraints in $\Phi$ are satisfied.
Otherwise, no assignment satisfies more than $2^{-\alpha
\ell}$-fraction of the constraints in $\Phi$.
\end{corollary}


\section{The Reduction}\label{sec:reduction}

Let $\eps>0$ be a constant from Theorem~\ref{thm:maintheorem}.
Recall that we want to prove $(2-\eps)$-hardness of approximation.


\paragraph{Overview} Starting with an instance $\phi$ of {\sf Max
  3SAT(5)}, we first perform the two-prover protocol with 
$\ell=\ceil{\log(3/\eps)/\alpha}$ rounds, and we enumerate all
possible constraints in $\Phi$. Next we transform $\Phi$ to an
instance of the Stackelberg problem in two steps, as follows. In the
first step of the reduction, we order the constraints in $\Phi$ to
get a {\em $(\delta, \gamma)$-far sequence} (see
Section~\ref{sec:farsequence}). In the second step, we convert such
sequence to an instance of the Stackelberg problem, denoted by $G$,
using the construction explained in Section~\ref{sec:construction}.


%
%

\subsection{Obtaining $(\delta, \gamma)$-far sequence}\label{sec:farsequence}
\begin{definition}($(\delta, \gamma)$-far constraint sequence)
Consider a sequence of all possible constraints $r_1,\ldots, r_M$ in
$\Phi$. 
A constraint $r_i$ is said to be {\em $\delta$-far} if for every $j:
i < j \leq i+ \ceil{\delta M}$, $q_1(r_i) \neq q_1(r_j)$ and
$q_2(r_i) \neq q_2(r_j)$. The sequence $r_1,\ldots, r_M$ is said to
be {\em $(\delta, \gamma)$-far} if at least $(1-\gamma)$-fraction of
constraints is $\delta$-far.
\end{definition}

We can obtain $(\delta, \gamma)$-far sequence with the right
parameter for our purpose using probabilistic
arguments. 

\begin{theorem}\label{theorem: far sequence}
For any $\ell\geq 1$, $\delta > 1/M$ and $\gamma\geq
(8\delta)5^\ell$, there is a polynomial-time algorithm $\mathcal{A}$
that outputs a $(\delta, \gamma)$-far sequence.
\end{theorem}
\begin{proof}
We present a randomized algorithm here. In Appendix, we derandomize
it to the desired $\mathcal A$ by the method of conditional
expectation.
Let $r_1, r_2, \ldots, r_M$ be the constraints. Let $\mathcal A'$ be
an algorithm that picks random a permutation $\pi: [M] \rightarrow
[M]$. We claim that the sequence $r_{\pi(1)}, \ldots, r_{\pi(M)}$ is
$(\delta,\gamma)$-far with probability at least $1/2$.
%

To prove the above claim, consider each constraint $r_i$. 
Let $J=\{j \in [M]: q_1(r_j)= q_1(r_i) \mbox{ or } q_2(r_j) =
q_2(r_i)\}$. Notice that $\abs{J}\leq 3^{\ell} + 5^{\ell} < 2 \cdot
5^{\ell}$ because there are $3^{\ell}$ constraints $r_j$ in $\Phi$
with $q_1(r_j)=q_1(r_i)$ and $5^{\ell}$ constraints $r_j$ in $\Phi$
with $q_2(r_j)= q_2(r_i)$. For each such $j \in J$, the probability
that $\abs{\pi(i)- \pi(j)} \leq \ceil{\delta M}$ is at most
$2\delta$. By applying the Union bound for all such $j \in J$, the
probability that $r_{\pi(i)}$ is {\em not} $\delta$-far is at most
$(4 \delta) 5^{\ell} \leq \gamma/2$. The expected number of
constraints that are not $\delta$-far is at most $\gamma M/2$, so by
Markov's inequality, the sequence is $(\delta,\gamma)$-far with
probability at least $1/2$, and the claim follows.
\end{proof}

\subsection{The Construction}
\label{sec:construction}

Given a ($\delta, \gamma)$-far sequence of constraints $r_1,
\ldots, r_M$, we construct an instance of {\sc StackSP} as
follows. For each constraint $r_i$, construct a gadget $G_i$
containing source $s_i$, destination $t_i$, and a set of
intermediate vertices $\set{u_i^a, v_i^a}_{a \in A(q_1(r_i))}$.
There are $2 \cdot 7^{\ell}$ such intermediate vertices (since
$\abs{A(q_1(r_i))}=7^\ell$).

Recall that, for each answer $a\in A(q_1(r_i))$, there exists a
unique consistent answer $a'\in A(q_2(r_i))$. In other words, for
each $a\in A(q_1(r_i))$ there exists a unique $a'\in A(q_2(r_i))$
such that $(a, a')$ satisfies the constraint $r_i$. From now on,
we will use $\pi_i$ to denote the function that maps each $a \in
A(q_1(r_i))$ to its consistent answer $a'\in A(q_2(r_i))$.
Therefore, each pair of $u_i^a, v_i^a$ corresponds to a pair of
possible answer $(a, \pi_i(a))$ that satisfies $r_i$.

Edges in each gadget $G_i$ are the following.

\bi

\item {\bf Fixed cost edges:} There is a fixed cost edge of
  cost $1$ from $s_i$ to $t_i$. There are also fixed cost edges of
  cost $0$ from $s_i$ to each of $u_i^a$, and from each of $v_i^a$ to
$t_i$.

\item {\bf Variable cost edges:} There is a variable cost edge
from $u_i^a$ to $v_i^a$ for each $a \in A(q_1(r_i))$.

\ei

Now we link all the gadgets together. First, for all $1 \leq i <
M$, we create a fixed cost edge of cost $0$ from $t_i$ to
$s_{i+1}$. We denote the source of instance $s=s_1$ and the sink
$t=t_M$ (i.e., we want to buy a shortest path from $s_1$ to
$t_M$).

Next, we add another set of fixed cost edges, called {\em
shortcuts}, whose job is to put constraints between pairs of edges
that represent inconsistent assignment. We only have shortcuts
between {\em far} gadgets. (Gadget $G_i$ is called a {\em far}
gadget if its corresponding constraint $r_i$ is a $\delta$-far
constraint.) Consider any pair of far constraints $r_i, r_j$ for
$i < j$ such that $r_i$ shares a query with $r_j$; i.e., either
$q_1(r_i) = q_1(r_j)$ or $q_2(r_i) = q_2(r_j)$. If $q_1(r_i) =
q_1(r_j)$, we add a shortcut from $v_i^{a_i}$ to $u_j^{a_j}$ for
every pair of $a_i\in A(q_1(r_i))$ and $a_j\in A(q_1(r_j))$ such
that $a_i \neq a_j$. For the case when $q_2(r_i)= q_2(r_j)$, we
add a shortcut from $v_i^{a_i}$ to $u_j^{a_j}$ for every pair of
$a_i, a_j$ such that $\pi_i(a_i) \neq \pi_j(a_j)$. We define the
cost of this shortcut to be $(j-i)/2$.

This completes the hardness construction. It is easy to see that
the instance size is polynomial (for completeness, we add the
proof in Appendix).






\section{Intuition and Overview of the Analysis}
\label{sec:intuition}

Before we move on to the analysis, we explain the intuition behind
the hardness construction in the previous section and the analysis
in the next section.

\paragraph{\NP-hardness}
First, let us understand what happens when we
apply the construction in Section~\ref{sec:construction} to Raz
verifier's $\Phi$ without applying Algorithm~$\mathcal A$ (cf.
Section~\ref{sec:farsequence}) to get a $(\delta, \gamma)$-far
sequence; in other words, the sequence of constraints is arbitrary.

We use the following example to convey the idea. Consider a \twosat
instance with three variables $x_1, x_2, x_3$ and two clauses
$C_1=(x_1\vee x_2)$ and $C_2=(x_1\vee x_3)$. (For the sake of
simplicity, we consider an instance of \twosat instead of {\sf Max
3SAT}.) The constraints of the Raz verifier with $\ell=1$ repetition
are $r_1=(C_1, x_1)$, $r_2=(C_1, x_2)$, $r_3=(C_2, x_3)$, and
$r_4=(C_2, x_1)$.
If we construct the graph $G$ from the sequence of constraints $r_1,
r_2, r_3, r_4$ according to the construction in
Section~\ref{sec:construction} then we will get the graph $G$ as in
Figure~\ref{fig:intuition-example}.

\begin{figure}
\includegraphics[width=\linewidth]{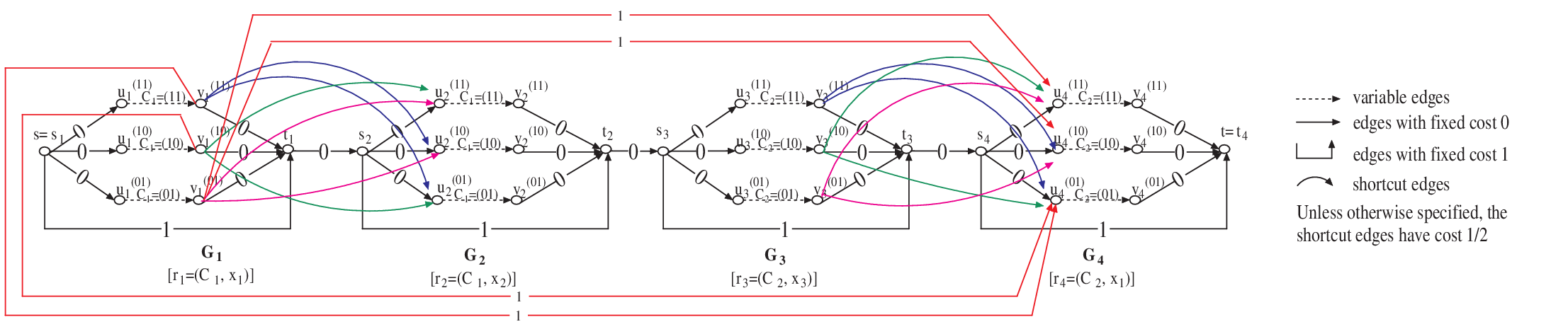}
\caption{\footnotesize Example of graph $G$ constructed from \twosat
$(x_1\vee x_2)\wedge (x_1\vee x_3)$ with $\ell=1$ repetition. Each
gadget $G_i$ is noted with the corresponding constraints $r_i$ and
each variable edge $u_i^av_i^a$ is noted with the corresponding
answer from Prover~1. Note that the corresponding answer from
Prover~2 can be identified easily. (For example, an edge
$u_1^{(10)}v_1^{(10)}$ corresponds to assigning $x_1=1$ and $x_2=0$.
Therefore, Prover~2's corresponding answer for
$u_1^{(11)}v_1^{(11)}$ is $x_1=1$.) The bigger picture is in
Appendix.}\label{fig:intuition-example}
\end{figure}


Consider any pricing $p$ and let $P$ be the corresponding shortest
path from $s$ to $t$. We classify the shortcuts whose both endpoints
are in $P$ into two types, edges that are {\em contained} in $P$ and
edges that are {\em induced} by $P$, as follows.


\begin{definition} We say that $P$ {\bf contains} an
edge $e$ if $e$ is an edge on $P$, and we say that $P$ {\bf induces}
$e$ if $e$ is not an edge on $P$ but both end vertices of $e$ are on
$P$. We say that $P$ {\bf involves} $e$ if $P$ contains or induces
$e$.
\end{definition}



Observe that if $P$ involves no shortcuts then we can construct a
satisfying assignment from $P$. For example, a path $s_1
u_1^{(11)}v_1^{(11)}t_1 s_2u_2^{(11)}v_2^{(11)}t_2
s_3u_3^{(10)}v_3^{(10)}t_3 s_4u_4^{(10)}v_4^{(10)}t_4$ involves no
shortcuts and could be converted to an assignment $x_1=1$, $x_2=1$
and $x_3=0$. Conversely, a satisfying assignment of $\Phi$ can also
be converted to a solution (a price function) with respect to which
the corresponding shortest path involves no shortcut edges.
Moreover, observe that if $P$ involves no shortcuts then we can get
a revenue of $M$ by setting price of all variable edges to $1$ and
we always get a revenue less than $M$ otherwise. The following
observation follows: {\em $\Phi$ has a satisfying assignment if and
only if there is a solution that gives a revenue of $M$ in the
corresponding graph $G$.}
%
%
This observation, along with the reduction from {\sf Max 3SAT},
already lead to the \NP-hardness of {\sc StackSP}. This is in fact
the essential idea used in the previous hardness results
\cite{RochSM05, Joret08}.

\paragraph{Beyond \NP-hardness}
To extend the above idea to a
constant-hardness, we further observe an effect of the shortcuts on
the revenue. In particular, we observe that if there are many
``parts'' of the shortest path that either contain or induce too
many shortcuts then the revenue can be essentially at most $M/2$. To
be more precise, let us first make the following two observations.

First, observe that if $P$ contains shortcuts $e_1, e_2, ..., e_k$,
for some $k$, with costs $c_1, c_2, ..., c_k$ then we can collect
a revenue of at most $M-\sum_{i=1}^k c_i$ from $P$. This is
because there is a path of length $M$ from $s$ to $t$ and, for each
$i$, once edge $e_i$ with fixed cost $c_i$ is used, the revenue on
$P$ decreases by $c_i$. For example, the path $P_1=s_1
u_1^{(11)}v_1^{(11)} u_2^{(10)}v_2^{(10)}t_2
s_3u_3^{(10)}v_3^{(10)}u_4^{(11)}v_4^{(11)}t_4$ contains two
shortcuts $v_1^{11}u_2^{10}$ and $v_3^{10}u_4^{11}$ of cost of $1/2$
each. Therefore, any solution in which such path is the
corresponding shortest path gives a revenue of at most
$4-1/2-1/2=3$.

Secondly, consider when $P$ induces a shortcut edge $e'$ from gadget
$G_i$ to gadget $G_j$ with cost $c'$ {\it and}, for some reason, the
edges in the gadgets $G_i$ and $G_j$ can have price at most $1$
each. Then we can collect a revenue of roughly $M - (j-i)+c'+2$.
This is because we cannot collect more than $c'+2$ on the subpath of
$P$ from gadget $G_i$ to gadget $G_j$. For example, consider a
path $P_2=s_1 u_1^{(11)}v_1^{(11)}t_1 s_2u_2^{(11)}v_2^{(11)}t_2
s_3u_3^{(01)}v_3^{(01)} t_3 s_4u_4^{(01)}v_4^{(01)}t_4$ which induces
a shortcut $v_1^{(11)}u_4^{(01)}$ of cost $1$. For a pricing that
$P_2$ is the shortest path, we can collect a revenue of at most $3$ for
the following reason. First, we can collect at most $1$ from edge
$u_1^{(11)}v_1^{(11)}$ because edge $s_1t_1$ would be used
otherwise. Similarly, we can collect at most $1$ from edge
$u_4^{(01)}v_4^{(01)}$. Moreover, we can collect at most $1$ from
$u_2^{(11)}v_2^{(11)}$ and $u_3^{(01)}v_3^{(01)}$ altogether because
the shortcut $v_1^{(11)}u_4^{(01)}$ would be used otherwise.

In summary, the observations above imply that a shortcut from
gadget $i$ to gadget $j$ (either contained or induced) causes the
revenue on the subpath from gadget $G_i$ to gadget $G_j$ to be
bounded by $(j-i)/2+2$.

\paragraph{The role of $(\delta, \gamma)$-far sequence}
Before we proceed to show the consequence of these observations, we
would like to eliminate the effect of the the constant ``+2'' in the
bound of the revenue above since it will be an obstacle in the
analysis. In particular, to get the factor of $2$ hardness, we would
like to say that we can get a revenue of roughly $(j-i)/2$ and somehow
conclude that the graph reduced from \ni gives a revenue of at
most $M/2$. (Recall that we can get a revenue of $M$ in \yi.)
However, the constant +2 is a problem when $j-i$ is small.

We eliminate the above effect in a straightforward way: instead of
including the shortcuts for every constraint, we consider only the
shortcuts with large cost $(j-i)/2$. The problem is, when we throw
away some constraints, the constraint satisfaction problem becomes
easier, and we should be able to satisfy more fraction of the
constraints. We do not want this to happen. We want to somehow make
sure that by neglecting a particular set of ``bad'' constraints, the
soundness parameter does not grow by much. Roughly speaking,
Section~\ref{sec:farsequence} shows that we can get the desired
properties while the soundness parameter remains comparatively
small. In particular, we lose an additive factor of $\gamma$ in the
soundness parameter. (Please refer to Section~\ref{sec:farsequence}
for more details.)


\paragraph{Getting 2-approximation hardness}
Now that we can eliminate the effect of the constant +2, let us see
how we can use the above two observations to conclude the
2-approximation hardness.
Intuitively, the two observations above imply that if the shortest
path $P$ involves many shortcuts then the revenue we can collect on
$P$ is essentially at most $M/2$. To prove this intuitive assertion,
we argue in the next section that we can always decompose $P$ into
three types of paths -- paths that look like $P_1$, paths that look
like $P_2$ and paths that can be converted to the solution for
$\Phi$ such that the number of satisfied constraints is equal to the
number of variable cost edges in such paths altogether. This
decomposition needs to be carefully designed to maintain the
properties of the three types of paths and will be elaborated in
Section~\ref{sec:path decomposition}.

Using the above decomposition and the fact that paths of the first
two types give a revenue of at most half of their lengths, we
conclude that the revenue is at most $M/2+c$ where $c$ is the
number of edges in the paths of the third type. Using the fact
that $\Phi$ is $(\delta, \gamma)$-far, we conclude that $c$ is at
most $(\gamma+\eps/3)M$ where $\eps$ is the constant as in
Theorem~\ref{thm:maintheorem}. By considering large enough $n$
(and thus, large enough $\abs{\Phi}$) and choosing an appropriate
value of $\delta$ and $\gamma$ so that $c\leq \eps M$, we have
that the revenue is at most $(1/2+\eps)M$. This implies the gap of
$2-\eps$, and Theorem~\ref{thm:maintheorem} thus follows. We
formalize these ideas in the next section.

\section{Analysis}\label{sec:analysis}

Now we prove Theorem~\ref{thm:maintheorem} using the reduction in
Section~\ref{sec:reduction}. Recall that $\eps$ is a constant as
in Theorem~\ref{thm:maintheorem} and we let
$\ell=\ceil{\log(3/\eps)/\alpha}$ (where $\alpha$ is as in
Theorem~\ref{theorem: Raz}), $\delta = (\eps/10) 5^{-\ell}$ and
$\gamma = \eps/3$. It follows that the soundness parameter of the
Raz verifier is $2^{-\alpha\ell}\leq \eps/3$. (I.e., if $\phi$ is
a \ni, then at most $\eps/3$ fraction of constraints in $\Phi$ can
be satisfied.)

In this section, we show that when the size of $\phi$ (denoted by
$n$) is large enough, the reduction gives a $(2-\eps)$-gap between
the case when $\phi$ is satisfiable and when it is not. In
particular, in section~\ref{subsection: yes instance}, we show that
if $\phi$ is satisfiable, then there is a price function that
collects a revenue of $M$. Moreover, in Section~\ref{subsection: no
instance} we show that if $\phi$ is not satisfiable and $n$ is large
enough, there is no pricing strategy which collects a revenue of
more than $(1/2+\eps)M$. The value of $n$ will be specified in
Section~\ref{subsection: no instance}.

\subsection{\yi}
\label{subsection: yes instance}

Let $f: \qset_1 \rightarrow \aset_1, \qset_2 \rightarrow \aset_2$
be an assignment that satisfies every constraint in $\Phi$. For
gadget $G_i$ corresponding to the variable $r_i$, set price $1$ to
the edge from $u^{a}_{i}$ to $v^{a}_{i}$ for $a= f(q_1(r_i))$.
Other variable cost edges in $G_i$ are assigned the price of
$\infty$. We now show that we can collect a revenue of $M$ in this
case.

Let $P$ be the shortest path on this graph with respect to the
above pricing. Notice that path $P$ does not contain any shortcut
since a shortcut only goes between two edges that represent
inconsistent assignments. (I.e., if there is a shortcut from
$v^{a_i}_{i}$ to $u^{a_j}_{j}$ on $P$ then either $a_i$ is not
consistent with $a_j$ or $\pi_i(a_i)$ is not consistent with
$\pi_j(a_j)$. Specifically, either $q_1(r_i)=q_1(r_j)$ and
$a_i\neq a_j$, or $q_2(r_i)=q_2(r_j)$ and $\pi_i(a_i)\neq
\pi_j(a_j)$. However, this is impossible since if
$q_1(r_i)=q_1(r_j)$ then $a_i=a_j=f(q_1(r_i))$ and, similarly, if
$q_2(r_i)=q_2(r_j)$ then $\pi_i(a_i)=\pi_j(a_j)=f(q_2(r_i))$.)
%

Since the shortcut is not used, the length of $P$ is exactly $M$.
Moreover, observe that the path that uses all variable edges of
price 1 also has length $M$. This path is a shortest path and
gives a total revenue of $M$.
%

\subsection{\ni}
\label{subsection: no instance}

We assume for contradiction that there is a pricing function which
collects a revenue of $(1/2+\eps) M$. Let $p$ be such pricing
function and let $P$ be the corresponding shortest path. Our goal
is to construct an assignment that satisfies more than $\eps M/3$
constraints in $\Phi$. This will contradict the soundness
parameter $\eps/3$ of the Raz verifier.



\begin{definition}
A subpath $Q \sse P$ is said to be a {\bf source-sink subpath} of
$P$ if it starts at some source $s_i$ and ends at some sink $t_j$
for $i \leq j$. For any source-sink subpath $Q$, denote by $s(Q)$
and $t(Q)$ the gadget index to which the source and sink of $Q$
belong respectively.

Now, let $Q$ be any source-sink subpath and let $s_i$ and $t_j$ be
its source and sink, respectively.
Let $\sset = \set{Q_1,\ldots, Q_k}$ be a set of source-sink
subpaths of path $Q$. We say that $\sset$ is a {\bf source-sink
partition} of path $Q$ if $s(Q_1) = i$, $t(Q_k)= j$, and for all
$p < k$, we have $t(Q_p) +1 = s(Q_{p+1})$.
\end{definition}


The following theorem is the key idea to proving the result.

\begin{theorem}[Path Decomposition]
\label{theorem: path decomposition} Let $p: E_v\rightarrow R^+\cup
\set{0}$ be the optimal pricing of the variable edges and $P$ be the
corresponding shortest path in the graph. Then we can find sets
$\rset$ and $\rset'$ such that the following properties hold.

\renewcommand{\theenumi}{D\arabic{enumi}}

\be

\item \label{property: partition} $\rset \cup \rset'$ is a
source-sink partition of $P$.

\item \label{property: small revenue in R'} The total revenue
collected from edges on paths in $\rset'$ is at most $M/2 +
O(1/\delta)$. In other words, $\sum_{e\in E_v\cap (\bigcup_{P\in
\rset'} P)} p(e)\leq M/2+O(1/\delta)$.

\item \label{property: small edge revenue} The price of any variable cost edge
in $\rset$ is at most $1$. That is, $p(e)\leq 1$ for any $e\in
E_v\cap (\bigcup_{P\in \rset} P)$.

\item \label{property: no shortcut} There is no shortcut between
any two variable cost edges in $\rset$.

\ee

\end{theorem}

We defer the proof of this theorem to the next section. Meanwhile
we show how the theorem implies that we can construct an
assignment that satisfies more than $\eps/3$ fraction of the
constraints in $\Phi$, thus a contradiction to the soundness
parameter. First, we consider only when $n$ is sufficiently large
so that we can collect at most $M/2+ O(1/\delta) < M/2+ \eps M/3$
from edges in $\rset'$ (from Property~\ref{property: small revenue
in R'}). Consequently, at least $2 \eps M/3$ must be collected
from edges in $\rset$.

Let $E'$ be the set of all variable cost edges that lie on some
paths in $\rset$. From Property~\ref{property: small edge revenue},
we have $\abs{E'} \geq 2\eps M/3$. Let $F \sse E'$ be the set of
edges in $E'$ that lie in far gadgets. Recall that we have at most
$\eps M/3$ gadgets that are not far (after we run an
algorithm~$\mathcal A$ in Theorem~\ref{theorem: far sequence}), so
$\abs{F} \geq \eps M/3$.

We are now ready to describe how we get an assignment that satisfies
a large fraction of constraints in $\Phi$. For each edge $e \in F$,
edge $e$ can be written as $u^{a_i}_{i}v^{a_i}_i$ for some gadget
$i$. We assign the answer $a$ for query $q_1(r_i)$ and $\pi_i(a)$
for query $q_2(r_i)$. This assignment satisfies the constraint
$r_i$. This process satisfies at least $\eps M/3$ constraints
corresponding to the edges in $F$ provided that there is no conflict
in assignment.

We argue that there is no such conflict since there is no shortcut
between the edges in $F$. I.e., assume that the above process
creates a conflict assignment to the same query $q$. This means that
there are two constraints $r_i, r_j \in \Phi$ for $i <j$ with
$q=q_1(r_i) =q_1(r_j)$ or  $q=q_2(r_i) =q_2(r_j)$ and such query $q$
was assigned different answers $a_i$ and $a_j$ when processing
gadgets $i$ and $j$. Since both $r_i$ and $r_j$ are far gadgets, by
construction, there must be a shortcut between two vertices
$v_i^{a_i}$ and $u_j^{a_j}$. This contradicts the fact that there is
no shortcut in $\rset$.

\subsection{Proof of Theorem~\ref{theorem: path
decomposition}}\label{sec:path decomposition} Consider any
source-sink subpath $Q$. Since there is a fixed-cost path of
length $t(Q) - s(Q) +1$ from $s_{s(Q)}$ to $t_{t(Q)}$, the revenue
collected on $Q$ is at most $t(Q)-s(Q)+1$, which will be denoted
by $\len(Q)$. We let $\rev(Q)$ be the revenue collected on subpath
$Q$, i.e. $\rev(Q) = \sum_{e\in Q\cap E_v} p(e)$. First, observe
the following lemma whose proof is simple and is deferred to
Appendix.

\begin{lemma}
\label{lemma: length of paths} If $\sset=\set{Q_1,\ldots, Q_k}$ is
a source-sink partition of $Q$, then $\sum_{j=1}^k \len(Q_j) =
\len(Q).$
\end{lemma}

We now explain the decomposition of the shortest path $P$ (from
Theorem~\ref{theorem: path decomposition}) into several source-sink
subpaths. Each subpath is contained in one of the sets $\rset$,
$\sset$ and $\tset$. In the end, we let $\rset'$ in the
Theorem~\ref{theorem: path decomposition} equal to $\tset \cup
\sset$. The composition consists of two phases. We next describe
each phase and prove the properties in Theorem~\ref{theorem: path
decomposition} along the way.

In the first phase, our goal is to make sure that $\rset$ contains
only source-sink subpaths that do not contain any shortcut.
Initially, we set $\rset$, $\sset$, and $\tset$ to $\rset=\{P\}$,
and $\sset=\tset=\emptyset$. We then remove the portion of paths
$P$ which contains the shortcut edges and add them to set $\sset$.
We ensure that paths are always cut into source-sink subpaths. In
particular, we do the following.

\paragraph{Phase 1:} Initially, $\rset= \set{P}$ and $\tset = \sset =
\emptyset$.
While there exists a path $P'\in \rset$ that contains a shortcut
edge, do the following. Let $vv'$ be any shortcut edge. Let $s_i$
be the last source vertex that appears before $v$ in $P'$ and let
$t_j$ be the first sink vertex that appears after $v'$ in $P'$. We
note that $i, j$ denote the gadget indices to which the vertices
belong. First remove $P'$ from $\rset$. Denote by $Q$ the
source-sink subpath of $P'$ from $s_i$ to $t_j$. We break $P'$
into three (possibly empty) source-sink subpaths $Q_l$, $Q$, and
$Q_r$; (i) $Q_l$ starts at $s(Q)$ and ends at vertex $t_{i-1}$,
(ii) $Q$ starts and ends at $s_i$ and $t_j$, respectively, and
(iii) $Q_r$ starts at $t_{j+1}$ and ends at $t(Q)$.
We then add $Q$ to $\sset$ and add $Q_l,Q_r$ back to $\rset$.
\\

Consider the set $\rset' = \sset \cup \tset$. We show that, after
this phase, the output satisfies properties~\ref{property:
partition}, \ref{property: small revenue in R'}, and
\ref{property: small edge revenue}. After the second phase,
property~\ref{property: no shortcut} will be satisfied while other
properties remain to hold. Observe that property~\ref{property:
partition} holds simply because the way we break path $P'$
guarantees that $s(Q_l)= s(P')$, $t(Q_l)+1 = s(Q)$, $t(Q)+1 =
s(Q_r)$, and $t(Q_r) = t(P')$. The next two lemmas prove
properties~\ref{property: small edge revenue} and \ref{property:
small revenue in R'}.

\begin{lemma}[Property \ref{property: small edge revenue}]
\label{lemma: limited revenue per gadget} After Phase 1, $p(e)\leq
1$ for any variable edge $e \in E_v$ that belongs to some path $Q$
in $\rset$.
\end{lemma}
\begin{proof}
Since path $Q$ does not contain shortcuts, vertices $s_i$ and
$t_i$ lie on $Q$ for all $s(Q)\leq i \leq t(Q)$. Recall that edge
$e$ can be written in the form $u^a_j v^a_j$ for some $j$ and $a
\in A(q_1(r_i))$. If $p(e)
>1$, we can obtain a path shorter than $P$
by using the fixed cost edge $s_j t_j$ of cost $1$ instead of $s_j
u^a_j v^a_j t_j$. This contradicts the fact that $P$ is a
shortest path.
\end{proof}

\begin{lemma}[Property~\ref{property: small revenue in R'}] After the
  first phase, the revenue in $\rset'=\sset\cup \tset$ is at most
  $M/2+ O(1/\delta)$. In particular, $\sum_{Q \in \sset} \rev(Q) \leq  \frac{1}{2}\paren{\sum_{Q \in
    \sset} \len(Q)} + O(1/\delta)$.
\end{lemma}

\begin{proof}
We will need the following claim.

\begin{claim}
\label{claim: limited revenue in T} For each path $Q \in \sset$,
we have $\rev(Q) \leq (\len(Q)+1)/2$.
\end{claim}
\begin{proof}
Consider path $Q \in \sset$ from $s_i$ to $t_j$. Recall that there
is a path of length $\len(Q)$ in $G$ from $s_i$ to $t_j$, so the
total cost of $Q$ is at most $\len(Q)$. It is, therefore, sufficient
to prove that the total cost of the shortcuts contained in $Q$ is at
least $(\len(Q)-1)/2$. The way we construct paths in $\sset$
guarantees that path $Q$ must be of the form
\[s_i \rightarrow u^{a_1}_{i_1}\Rightarrow v^{a_1}_{i_1} \rightarrow
u^{a_2}_{i_2}\Rightarrow v^{a_2}_{i_2} \rightarrow \ldots
\Rightarrow v^{a_q}_{i_q} \rightarrow t_j\]

where $i_1= i, i_q= j$, and edges of the form $u^{a_x}_{i_x}
\Rightarrow v^{a_{x}}_{i_{x}}$ are the variable cost edges, from
which we can collect a revenue. Other edges of the form $v^{a_x}_{i_x}
\rightarrow u^{a_{x+1}}_{i_{x+1}}$, for $1 \leq x <q$, are
shortcuts. Hence the total cost of shortcuts can be written as a
telescopic sum, $\sum_{x=1}^{q-1}
\paren{\frac{i_{x+1}-i_x}{2}} = (j-i)/2=(\len(Q)-1)/2$.
\end{proof}

By the claim, $\sum_{Q \in \sset} \rev(Q) \leq \sum_{Q \in \sset}
\paren{\len(Q)/2 + 1/2} \leq \frac{1}{2}\paren{\sum_{Q \in \sset}
\len(Q)}+ \abs{\sset}/2$. It then suffices to bound the size of
set $\sset$ by $O(1/\delta)$. Notice that each path in $\sset$
contains at least one shortcut. Recall that, by the construction
(cf. Section~\ref{sec:construction}), each shortcut only goes from
$v^a_i$ to $u^{a'}_j$ if $\abs{j-i} \geq \delta M$. Since the
intervals in the set $\set{[s(Q), t(Q)]: Q \in \sset}$ are
disjoint (by definition of source-sink partition), we can have at
most $O(1/\delta)$ paths in $\sset$.
\end{proof}


This completes the description and the proof of Phase~1. Now every
path in $\rset$ contains no shortcut. In phase~2, our goal is to
eliminate the shortcuts between paths in $\rset$. (Note that these
shortcuts are not contained in $P$.) Roughly speaking, we scan the
gadgets from left to right and once we find such shortcut, we move
the whole path that induces this shortcut to the set $\tset$. The
detail is as follows.

\paragraph{Phase 2:} Initially, we have $\rset$ and $\sset$ from Phase
1, and $\tset =\emptyset$. We proceed in iterations starting from
iteration $1$. The description of iteration $i$ is as follows:
%

\bi

\item We first check if source $s_i$ belongs to some path in
$\rset$. If not, we proceed to iteration $i+1$.

\item If $s_i$ does belong to any path $Q$ in $\rset$, we do the following.
We check if there is a shortcut (that is not contained in $Q$)
leaving from some vertex $v_i^{a_i}$ on $Q$ to some vertex
$u_j^{a_j}$ on some path $Q' \in \rset$. Note that $Q$ and $Q'$ may
be the same. Let $P' \sse P$ be the source-sink subpath from $s_i$
to $t_j$. We first remove from $\rset$ and $\sset$, all paths $Q''$
such that $Q'' \cap P' \neq \emptyset$. Let $Q_l$ be the source-sink
subpath of $Q$ with $s(Q_l) = s(Q)$ and $t(Q_l) = s(P')-1$. Also, we
let $Q_r$ be the source-sink subpath of $Q'$ with $s(Q_r) = t(P')+1$
and $t(Q_r) = t(Q')$. We add $P'$ to $\tset$, and add $Q_l$ and
$Q_r$ back to $\rset$.




\ei

We now check the properties.
Property~\ref{property: partition} holds simply because, in each
iteration, we remove only subpaths of what we will add (i.e., we may
add paths $Q$, $P'$ and $Q'$ to $\rset$ and $\tset$ and remove only
subpaths of $Q\cup P'\cup Q'$). Since paths in $\rset$ only get
chopped off, Lemma~\ref{lemma: limited revenue per gadget} still
holds, and so does property~\ref{property: small edge revenue}.
Properties~\ref{property: no shortcut} and \ref{property: small
revenue in R'} follow from the following Lemmas whose proofs are in
Appendix.


\begin{lemma}[Property~\ref{property: no shortcut}]
\label{lemma: no shortcut} After Phase 2, there is no shortcut
between any two subpaths in $\rset$.
\end{lemma}


\begin{lemma}[Property~\ref{property: small
revenue in R'}] $\sum_{Q \in \tset} \rev(Q) \leq
\frac{1}{2}\paren{\sum_{Q \in \tset} \len(Q)} + O(1/\delta).$
\label{lemma:small revenue in T}
\end{lemma}
%
%
%



\nu{To do for the final version: Acknowledgement.}

  \let\oldthebibliography=\thebibliography
  \let\endoldthebibliography=\endthebibliography
  \renewenvironment{thebibliography}[1]{%
    \begin{oldthebibliography}{#1}%
      \setlength{\parskip}{0ex}%
      \setlength{\itemsep}{0ex}%
  }%
  {%
    \end{oldthebibliography}%
  }
{ \small
\bibliographystyle{plain}
\bibliography{stackelberg-hardness}
}

\newpage \appendix
\section*{APPENDIX}
\section{Derandomization of Algorithm~$\mathcal A'$ in Theorem~\ref{theorem: far sequence}}


Now we derandomize $\mathcal A'$ to get a deterministic algorithm
$\mathcal A$ by the method of conditional expectation. Let $Y$
denote the number of constraints that are not $\delta$-far with
respect to a random permutation $\pi$. For a fixed permutation
$\pi'$, let $\events(\pi',I)$ be the event that $\pi$ agrees with
$\pi'$ on set $I$ (i.e., $\pi'(i)=\pi(i)$ for all $i\in I$).
Notice that, we can efficiently compute $\expect{}{Y \mid
\events(\pi',I)}$ for any $\pi'$ and $I$ where the expectation is
over random permutation $\pi$. Therefore, for $i=1, 2, ...$, we
deterministically pick the value of $\pi'(i)$ that maximizes the
value of $\expect{}{Y \mid \events(\pi',\set{1,\ldots, i-1})}$.

\section{Construction Size}

We first calculate the size of each
gadget $G_i$. There are $O(7^{\ell})$ vertices and $O(7^{\ell})$
edges for each gadget. Next, we count the number of shortcuts. For
each pair of constraints $r_i$ and $r_j$, there are at most
$O(7^{2\ell})$ shortcuts between their intermediate vertices. Since
there are $(5n)^{\ell}$ gadgets, the graph size is at most
$O(n)^{O(\ell)}$. Since $\ell=\ceil{\frac{\log(3/\eps)}{\alpha}}$,
the construction size is $O(n)^{O(1/\eps)}$ which is polynomial in
$n$ if $\eps$ is a constant.

\section{Omitted Proofs from Section~\ref{sec:analysis}}

\subsection{Proof of Lemma~\ref{lemma: length of paths}}
$\sum_{j=1}^k \len(Q_j) = \sum_{j=1}^k (t(Q_j)-s(Q_j)+1) =
t(Q_k)-s(Q_1) +1= t(Q)+1-s(Q) = \len(Q)$ where the second equality
is because $t(Q_j)+1=s(Q_{j+1})$ for all $j\leq k$ and the third
equality is because $t(Q_k)=t(Q)$ and $s(Q_1)=s(Q)$.\bun{We may
have to rearrage this proof for easy reading}

%
%

\subsection{Proof of Lemma~\ref{lemma: no shortcut}}
Notice that once a shortcut leaving gadget $i$ is found, the whole
part of gadget $i$ is removed completely from $\rset$. Therefore,
after iteration $i$, there is no shortcut leaving the vertex in $P
\cap G_i$ to other vertices lying on some path in $\rset$. (In
fact, the vertex in $P\cap G_i$ is not in any path in $\rset$
anymore.)


\subsection{Proof of Lemma \ref{lemma:small revenue in T}}

Similarly to Claim~\ref{claim: limited revenue in T}, we can also
bound the revenue on paths in $\tset$ as summarized in the following
claim whose proof can be found in Appendix.

\begin{claim}
\label{claim: limited revenue in S} For each path $Q \in \tset$, we
have $\rev(Q) \leq \frac 1 2 \len(Q) + 2$
\end{claim}
\begin{proof}
Consider path $Q \in \tset$ from $s_i$ to $t_j$. Path $Q$ can be
written in the form: \[s_i \rightarrow u^{a_i}_i \Rightarrow
v^{a_i}_i \rightarrow \ldots \rightarrow u^{a_j}_j \Rightarrow
v^{a_j}_j \rightarrow t_j.\] Note that we do not assume any
structure of the path from $v^{a_i}_i$ to $u^{a_j}_j$. Also, recall
that edges $u^{a_i}_iv^{a_i}_i$ and $u^{a_j}_jv^{a_j}_j$ were in
$\rset$ after Phase~1 and moved to $\tset$ in Phase~2. Moreover,
there is a shortcut edge from $v^{a_i}_i$ to $u^{a_j}_j$ (which is
not in $Q$).

Now, let $Q'$ be the subpath of $Q$ from $v^{a_i}_i$ to $u^{a_j}_j$,
and $e_i, e_j$ be the edges $u^{a_i}_i v^{a_i}_i$ and $u^{a_j}_j
v^{a_j}_j$, respectively. Then $Q = s_i e_i Q' e_j t_j$. The revenue
collected on $Q$ comes from edges in $Q'$ and $e_i$ and $e_j$. Since
both $e_i$ and $e_j$ belonged to some paths in $\rset$ after Phase
1, we have $p(e_i) + p(e_j) \leq 2$ (cf. Lemma~\ref{lemma: limited
revenue per gadget}). Path $Q'$ can collect revenue of at most
$(j-i)/2$ due to the fact that there is a shortcut edge $v^{a_i}_i
u^{a_j}_j$ of cost $(j-i)/2$. Overall, the revenue on $Q$ is at most
$(j-i)/2+2 < \frac{1}{2}\len(Q) + 2$.
\end{proof}

Since every path $Q \in \tset$ induces some shortcut edges (i.e.,
there is a shortcut edge between some pairs of vertices in $Q$), the
length of such path is at least $\delta M$. Therefore, $\abs{\tset}
\leq O(1/\delta)$. We apply Claim~\ref{claim: limited revenue in S}
for every path in $\tset$ and sum them up. This immediately gives
the lemma.


\begin{figure}
\begin{center}
\includegraphics[angle=90, height=\textheight]{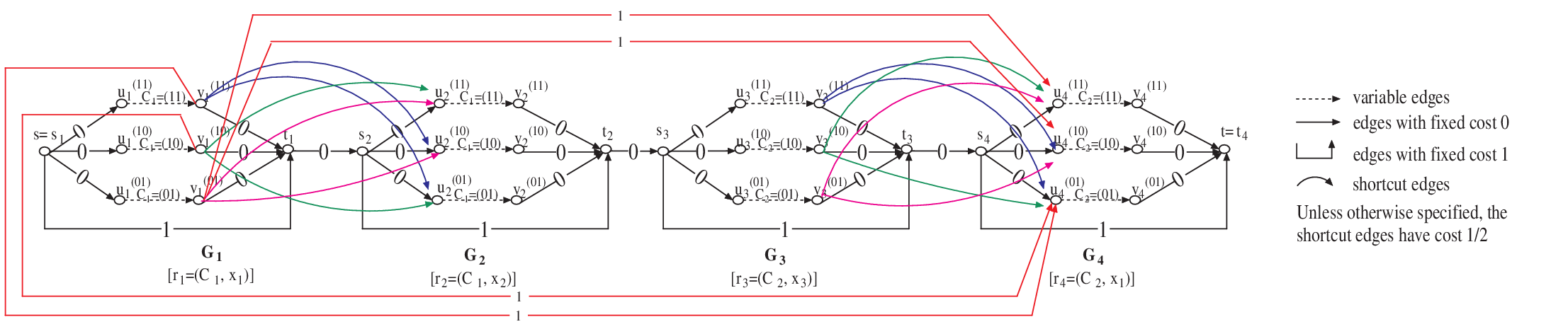}
\end{center}
\caption{\footnotesize Example of graph $G$ constructed from \twosat
$(x_1\vee x_2)\wedge (x_1\vee x_3)$ with $\ell=1$ repetition. Each
gadget $G_i$ is noted with the corresponding constraints $r_i$ and
each variable edge $u_i^av_i^a$ is noted with the corresponding
answer from Prover~1. Note that the corresponding answer from
Prover~2 can be identified easily. (For example, an edge
$u_1^{(10)}v_1^{(10)}$ corresponds to assigning $x_1=1$ and $x_2=0$.
Therefore, Prover~2's corresponding answer for
$u_1^{(11)}v_1^{(11)}$ is $x_1=1$.) }
\end{figure}

\end{document}